\documentclass[11pt]{amsart}
\usepackage{fullpage}
\usepackage{amsthm,amsmath,amssymb}
\usepackage[foot]{amsaddr}

\usepackage{enumitem}
\usepackage{booktabs}
\usepackage{graphicx}
\usepackage{subcaption}
\usepackage{tikz}

\usepackage[colorlinks=true,linkcolor=red,breaklinks=true,citecolor=blue,urlcolor=cyan]{hyperref}

\graphicspath{{./figs/}} 

\newtheorem{theorem}{Theorem}[section]

\newtheorem{corollary}[theorem]{Corollary}
\theoremstyle{remark}

\title{Maximum Entropy Estimation of Heterogeneous Causal Effects}

\author{Brian Knaeble}
\email{bknaeble@uvu.edu}

\author{Mehdi Hakim-Hashemi}

\author{Mark A. Abramson}


\begin{document}

\begin{abstract}
For the purpose of causal inference we employ a stochastic model of the data generating process, utilizing individual propensity probabilities for the treatment, and also individual and counterfactual prognosis probabilities for the outcome. We assume a generalized version of the stable unit treatment value assumption, but we do not assume any version of strongly ignorable treatment assignment. Instead of conducting a sensitivity analysis, we utilize the principle of maximum entropy to estimate the distribution of causal effects. We develop a principled middle-way between extreme explanations of the observed data: we do not conclude that an observed association is wholly spurious, and we do not conclude that it is wholly causal. Rather, our conclusions are tempered and we conclude that the association is part spurious and part causal. In an example application we apply our methodology to analyze an observed association between marijuana use and hard drug use.
\end{abstract}

\maketitle

\section{Introduction}
We start with a thought experiment. Suppose one million coins are flipped and approximately $600,000$ come up heads. Here are two explanations.
\begin{enumerate}
\item Each coin has probability of heads equal to approximately $0.6$.
\item Half the coins have probability of heads equal to approximately $0.9$, and half the coins have probability of heads equal to approximately $0.3$.
\end{enumerate}
There are many other explanations. Which explanation is best?

The first explanation is reasonable and supported by the principle of maximum entropy. The principle of maximum entropy expresses a claim of epistemic modesty. The principle guides us to select a distribution based on known information. Suppose we know that the variance of the probabilities is $0.09$. In light of that knowledge the second explanation is good.

In this paper we will show how to utilize the principle of maximum entropy to estimate a distribution of causal effects. Our running example application is taken from an observational study \cite{PATH} that observed an association between marijuana use and hard drug use. Frequency counts are displayed in Table \ref{tab2}. We are interested in whether or not marijuana is a gateway drug. We are interested in whether the association between marijuana and hard drug use admits a causal interpretation, and specifically what kind of causal interpretation. We are asking the following question. Could marijuana use cause hard drug use for some individuals but not others?
\begin{table}[ht]
\centering
\caption{A $2\times 2$ contingency table showing an association (odds ratio OR$=25$) between Marijuana use and hard drug use.}
\label{tab2}
\begin{tabular}{rcc}
\toprule
 & \multicolumn{2}{c}{Marijuana}\\
\cmidrule{2-3}
Hard Drugs&No&Yes\\
\hline
Yes&81&796\\
No&4,201&1,680\\
\hline
\end{tabular}
\end{table}

If within our model of the data generating process we utilize propensity probabilities of exposure to marijuana, then by symmetry we should also utilize prognosis probabilities for the outcome of hard drug use, and those prognosis probabilities may depend on exposure. Even if we accept that marijuana causes hard drug use, a given individual is not destined nor determined to use hard drugs at the time when they first use marijuana. If the process leading to marijuana is stochastic, then the process leading to hard drug use is likely stochastic as well. Marijuana use may affect the parameters of a counterfactual distribution, so we need a mathematical framework for reasoning with stochastic counterfactuals. There is not a lot of literature about stochastic counteractuals \cite{VR}, but our work here is closely related to sensitivity analysis for causal inference. This current work builds upon \cite{KOA} and \cite{Knaeble2023}, where coefficients of determination summarize unmeasured sets of covariates. 

Here, our main contribution is the following: within a general, propensity-prognosis model of the data generating process, we will show how to utilize the principle of maximum entropy to estimate distributions of heterogeneous causal effects. We do so with measured covariate data and transported coefficients of determination. Additional contributions include mathematical results: one that expresses Tjur's Pseudo $R^2$ parameter \cite{Tjur2009} within our framework, and others that support principled inference of homogeneous causal effects. We also introduce a practical algorithm to estimate distributions of heterogeneous causal effects. Our notation is introduced in Section \ref{methods}. Mathematical results are described in Section \ref{results}. The details of our example application are described in Section \ref{applications}. A discussion occurs in Section \ref{discussion}.
\section{Methods}
\label{methods}
For individual $i$ we write $e_i$ to indicate exposure (or treatment) and $d_i$ to indicate a dichotomous outcome (or disease). We write $\pi_i$ for propensity probability of exposure, $r_{0i}$ for prognosis probability of the outcome in the absence of exposure, and $r_{1i}$ for prognosis probability of the outcome in the presence of exposure. For individual $i$ causation is present if $r_{0i}\neq r_{1i}$. Here is a formal summary. \begin{itemize}
\item $e_i\sim \textrm{Bernoulli}(\pi_i)$,
\item $d_i(e_i=0)\sim \textrm{Bernoulli}(r_{0i})$, and
\item $d_i(e_i=1)\sim \textrm{Bernoulli}(r_{1i})$.
\end{itemize}

In the language of \cite{Imbens2015} our processes are assumed to be individualistic and probabilistic, and we have assumed for each $i$ that $e_i$ does not depend on $(d_i(e_i=0),d_i(e_i=1))$ (see the latter requirement of SUTVA in \cite[Section 3.2]{Imbens2015}). We will utilize a vector of measured covariates $x$, but we will not assume that the processes are unconfounded, even conditional on $x$. In other words, we do not accept the assumption of strongly ignorable treatment assignment \cite[Section 3.4, Definition 3.6]{Imbens2015}. We measure $\{e_i,d_i,x_i\}_{i=1}^N$ while $\{\pi_i,r_{0i},r_{1i}\}_{i=1}^N$ are latent and unknown.

Within our model, information entropy, as defined in \cite{Shannon1948}, is given by the following mathematical function of Bernoulli parameter variables.
\begin{align*}
H(\pi,r_{0},r_{1})&=-\pi \log(\pi)-(1-\pi)\log(1-\pi)
\\&-(1-\pi)(r_0\log(r_0)+(1-r_0)\log(1-r_0))
\\&-\pi(r_1\log(r_1)+(1-r_1)\log(1-r_1))
\\&= -(1-\pi)r_0\log((1-\pi)r_0)-\pi r_1\log(\pi r_1)\\&-(1-\pi)(1-r_0)\log((1-\pi)(1-r_0))\\&-\pi(1-r_1)\log(\pi(1-r_1)).
\end{align*}

The following is our basic, maximum entropy, optimization problem.
\begin{subequations}
\label{basic}
\begin{align}
H_\star &= \max_{\{\pi_i,r_{0i},r_{1i}\}_{i=1}^N} \  \sum_{i=1}^N H(\pi_i,r_{0i},r_{1i}), \
\textrm{subject to}\\
\sum (1-\pi_i)r_{0i}/N \ &= \ P(e=0,d=1) \\
\sum \pi_i r_{1i}/N \ &= \ P(e=1,d=1) \\
\sum (1-\pi_i)(1-r_{0i})/N \ &= \ P(e=0,d=0) \\
\sum \pi_i(1-r_{1i})/N  \ &= \ P(e=1,d=0).
\end{align}
\end{subequations}
The reasoning behind our equality constraints is that we'd like for the heterogeneous parameters $\{\pi_i,r_{0i},r_{1i}\}_{i=1}^N$ to provide a consistent explanation for the observed data $\{e_i,d_i,x_i\}_{i=1}^N$. 

We take measured attributes of individuals into account by categorizing each of the $N$ individuals into $d$ categories, $x_1,x_2,...,x_c,...,x_d$, where $c$ is the indexing variable, and the corresponding numbers of individuals in each category are $N_1,N_2,...,N_c,...,N_d$ and satisfy $\sum_{c=1}^dN_c=N$. The categories should be defined from a covariate vector $x$ which takes values in a bounded set $X$. 
For each individual, covariate measurements should be made prior to exposure and at or prior to the time when $\pi$ is defined. 
We write $I_c$ as an indicator function that takes the value $1$ if an individual is in category $x_c$ and the value $0$ otherwise. Here is a maximum entropy problem that incorporates measured covariate data.  
\begin{subequations}
\label{c}
\begin{align}
H_\star &= \max_{\{\pi_i,r_{0i},r_{1i}\}_{i=1}^N} \  \sum_{i=1}^N H(\pi_i,r_{0i},r_{1i}), \textrm{~subject to},\\ 
\forall c\in \{1,\cdots,d\}, \sum I_c(1-\pi_i)r_{0i}/N_c \ &= \ P(e=0,d=1|x=x_c), \\
\sum I_c\pi_i r_{1i}/N_c \ &= \ P(e=1,d=1|x=x_c), \\
\sum I_c(1-\pi_i)(1-r_{0i})/N_c \ &= \ P(e=0,d=0|x=x_c), \\
\sum I_c\pi_i(1-r_{1i})/N_c  \ &= \ P(e=1,d=0|x=x_c).
\end{align}
\end{subequations}

Define the expected, individual risk, $r=\pi r_1+(1-\pi)r_0$. Define Tjur's pseudo $R^2$ as ``the difference between the averages of fitted values for successes and failures'' \cite{Tjur2009}. We write $R^2_{p;\pi}$ and $R^2_{p;r}$ for practical parameters of models reported in the literature, and we write $R^2_{t;\pi,e}$ and $R^2_{t;r,d}$ for theoretical parameters defined from an idealized model with fitted values equal to the true $\pi$ and $r$ values; see Theorem \ref{th2} and Corollary \ref{cor2}. Here is an expanded maximum entropy problem with two additional constraints:
\begin{subequations}
\label{cR}
\label{expanded}
\begin{align}
H_\star &= \max_{\{\pi_i,r_{0i},r_{1i}\}_{i=1}^N} \  \sum_{i=1}^N H(\pi_i,r_{0i},r_{1i}), \textrm{~subject to}
\\
\label{piineq}\sum (\pi_i-\bar{\pi})^2/N &\geq P(e=1)(1-P(e=1))R^2_{p;\pi},\\
\label{rineq}\sum (r_i-\bar{r})^2/N &\geq P(d=1)(1-P(d=1))R^2_{p;r}, \textrm{~and}\\
\forall c\in \{1,...,d\}, \sum I_c(1-\pi_i)r_{0i}/N_c \ &= \ P(e=0,d=1|x=x_c), \\
\sum I_c\pi_i r_{1i}/N_c \ &= \ P(e=1,d=1|x=x_c), \\
\sum I_c(1-\pi_i)(1-r_{0i})/N_c \ &= \ P(e=0,d=0|x=x_c), \\
\sum I_c\pi_i(1-r_{1i})/N_c  \ &= \ P(e=1,d=0|x=x_c).
\end{align}
\end{subequations}
To solve the maximization problem of (\ref{expanded}) we reformulate the problem with a change of variables resulting in a linear objective function and linear equality constraints. Any distribution of $\{\pi_i,r_{0i},r_{1i}\}$ values  is supported on the open unit cube \[C=\{(\pi,r_0,r_1):0<\pi,r_0,r_1<1\}.\] Given a large, natural number $m$ we define for $j,k,l=1,...m$ a subcube \[C_{j,k,l}=\{(\pi,r_0,r_1):(j-1)/m<\pi<j/m,(k-1)/m<r_0<k/m,(l-1)/m<r_1<l/m\}.\] We write $(\pi_j,r_{0k},r_{1l})$ as the centers of those subcubes. For each $c\in\{1,...,d\}$ the proportion of mass of a distribution on $C\times X$ that both lies within a subcube $C_{j,k,l}$ and has $x=x_c$ is expressed with a weight $w_{c,j,k,l}$, and the sum of all weights is required to be unity. We assume sufficiently large $N$ and $m$ values and allow fractional weights smaller than $1/N$. The maximization problem in (\ref{expanded}) becomes
\begin{subequations}
\label{transformed}
\begin{align}
&H_\star \approx \max_{\{(w_{c,j,k,l})\}_{j=1,k=1,l=1}^m} \\
&-\left(\sum_{c=1}^d\sum_{j=1}^m\sum_{k=1}^m\sum_{l=1}^m Nw_{c,j,k,l}(1-\pi_j)r_{0k}\log (1-\pi_j)r_{0k}\right)\\
&-\left(\sum_{c=1}^d\sum_{j=1}^m\sum_{k=1}^m\sum_{l=1}^m Nw_{c,j,k,l}\pi_j(r_{1l})\log \pi_j(r_{1l})\right)\\
&-\left(\sum_{c=1}^d\sum_{j=1}^m\sum_{k=1}^m\sum_{l=1}^m Nw_{c,j,k,l}(1-\pi_j)(1-r_{0k})\log (1-\pi_j)(1-r_{0k})\right)\\
&-\left(\sum_{c=1}^d\sum_{j=1}^m\sum_{k=1}^m\sum_{l=1}^mNw_{c,j,k,l}\pi_j(1-r_{1l})\log \pi_j(1-r_{1l})\right)\\
\nonumber &\textrm{subject to},\\
\label{fcineq}   
& \sum_{c=1}^d\sum_{j=1}^m\sum_{k=1}^m\sum_{l=1}^m w_{c,j,k,l}(\pi_j-P(e=1))^2\geq P(e=1)(1-P(e=1))R^2_{p;\pi}, \\
\label{scineq}& \sum_{c=1}^d\sum_{j=1}^m\sum_{k=1}^m\sum_{l=1}^m w_{c,j,k,l}(r_{k,l}-P(d=1))^2\geq P(d=1)(1-P(d=1))R^2_{p;r},\\
\nonumber &\textrm{and~} \forall c\in\{1,...,d\},\\
&\sum_{j=1}^m\sum_{k=1}^m\sum_{l=1}^mI_cw_{c,j,k,l}(1-\pi_j)r_{0k}   =  P(e=0,d=1|x=x_c), \\
& \sum_{j=1}^m\sum_{k=1}^m\sum_{l=1}^mI_cw_{c,j,k,l}\pi_j(r_{1l})    =  P(e=1,d=1|x=x_c), \\
& \sum_{j=1}^m\sum_{k=1}^m\sum_{l=1}^mI_cw_{c,j,k,l}(1-\pi_j)(1-r_{0k})    =  P(e=0,d=0|x=x_c), \\
& \sum_{j=1}^m\sum_{k=1}^m\sum_{l=1}^mI_cw_{c,j,k,l}\pi_j(1-r_{1l})    =  P(e=1,d=0|x=x_c).
\end{align}
\end{subequations}
Each line of the transformed problem (\ref{transformed}) is linear in its variables (the weights), and it can therefore be solved with linear programming, as long as the measured covariates are low-dimensional. The number of variables in the LP problem is $dm^3$. In our example application of Section \ref{applications} we have $d=10$, and we had good run times and decent accuracy with $40<m<95$, see Appendix \ref{AppB}. We implemented our algorithm in Matlab, and our code is available upon request.
\section{Results}
\label{results}
The results of this section support our formulations of the optimization problems of Section \ref{methods} and also the applications of Section \ref{applications}.
\begin{theorem}[Homogeneous Effects] 
\label{th1} The solution to the basic maximization problem in (\ref{basic}) is a single point mass distribution: for each and every individual
\[(\pi_i,r_{0i},r_{1i})=(P(e=1),P(d=1|e=0),P(d=1|e=1).\]
\end{theorem}
\begin{proof}
Define $p_{01}=P(e=0,d=1)$, $p_{11}=P(e=1,d=1)$, $p_{00}=P(e=0,d=0)$, and $p_{10}=P(e=1,d=0)$. The Lagrangian for the original problem with no variance constraints is given by
\begin{eqnarray*}
 L &=& \sum_{i=1}^N \left\{ \pi_i \log \pi_i + (1 - \pi_i) \log(1 - \pi_i) \right. \\
   &+&  (1 - \pi_i) \left[ r_{0i} \log r_{0i} + (1 - r_{0i}) \log(1-r_{0i}) \right] \\
   &+& \left.      \pi_i  \left[ r_{1i} \log r_{1i} + (1 - r_{1i}) \log(1-r_{1i}) \right] \right\} \\
   &+& \lambda_1 \left[ \sum_{i=1}^N (1 - \pi_i)    r_{0i}  - N p_{01} \right] + \lambda_2 \left[ \sum_{i=1}^N \pi_i    r_{1i}  - N p_{11} \right] \\
   &+& \lambda_3 \left[ \sum_{i=1}^N (1 - \pi_i) (1-r_{0i}) - N p_{00} \right] + \lambda_4 \left[ \sum_{i=1}^N \pi_i (1-r_{1i}) - N p_{01} \right].
\end{eqnarray*}

Taking partial derivatives and grouping terms, we obtain optimality conditions; namely, for all $i = 1, 2, \ldots, N$,
\begin{eqnarray*}
  \pi_i \log \pi_i - (1 - \pi_i) \log(1 - \pi_i) - r_{0i} \log r_{0i} - (1 - r_{0i}) \log(1-r_{0i}) && \\
   + r_{1i} \log r_{1i} + (1 - r_{1i}) \log(1-r_{1i}) && \\
   - \lambda_1 r_{0i} + \lambda_2 r_{1i} - \lambda_3 (1 - r_{0i}) + \lambda_4 (1 - r_{1i}) &=& 0 \\~\\
  (1 - \pi_i) \left[ \log r_{0i} - \log (1 - r_{0i}) + \lambda_1 - \lambda_3 \right] &=& 0 \\
       \pi_i  \left[ \log r_{1i} - \log (1 - r_{1i}) + \lambda_2 - \lambda_4 \right] &=& 0.
\end{eqnarray*}

Since $\pi_i \in (0,1)$, we can solve the second and third equations for $r_{0i}$ and $r_{1i}$, which gives
\begin{eqnarray*}
  r_{0i} &=& \dfrac{1}{1 + \exp(-\lambda_3/\lambda_1)}, \quad i = 1, 2, \ldots, N \\
  r_{1i} &=& \dfrac{1}{1 + \exp(-\lambda_4/\lambda_2)}, \quad i = 1, 2, \ldots, N.
\end{eqnarray*}
That is, $r_{0i}$ and $r_{1i}$ must have the same values for all $i$.  Similarly, substituting these expressions into the first optimality equation, we see by the same argument that $\pi_i$ have the same value for all $i$.  Then by substitution, the constraints become simply
\begin{eqnarray*}
  (1 - \pi_i)      r_{0i}  = p_{01}, && \pi_i      r_{1i}  = p_{11}, \\
  (1 - \pi_i) (1 - r_{0i}) = p_{00}, && \pi_i (1 - r_{1i}) = p_{10}.
\end{eqnarray*}
Adding the second and fourth equations yields $\pi_i = p_{11} + p_{10} = P(e = 1)$.  From here, we get
\begin{eqnarray*}
  r_{0i} = \dfrac{p_{01}}{1 - \pi_i} = \dfrac{p_{01}}{p_{01} + p_{00}}, && r_{1i} = \dfrac{p_{11}}{\pi_i}     = \dfrac{p_{11}}{p_{11} + p_{10}}.
\end{eqnarray*}
\end{proof}
\begin{corollary}[Conditional Homogeneous Effects]
\label{condhomo}
The solution to the maximization problem in (\ref{c}) consists of a mixture of single point mass distributions, each conditional on $x$: $\forall x_c\in X$, and for each individual with $x=x_c$, we have
\[(\pi_i,r_{0i},r_{1i})=(P(e=1|x=x_c),P(d=1|e=0,x=x_c),P(d=1|e=1,x=x_c).\]
\end{corollary}

\begin{theorem}[Proportion of Variation Explained]
\label{th2} Let $R^2_{t;\pi,e}$ denote Tjur's pseudo-$R^2$ parameter between $\pi$ and $e$ as defined in Section \ref{methods}. Let $\sigma^2_\pi$ be the variance of the latent $\pi$ values, and let $P(e=1)$ be the observed proportion of individuals with $e=1$. Then,
\[R^2_{t;\pi,e}=\frac{\sigma^2_{\pi}}{P(e=1)(1-P(e=1)}.\]
\end{theorem}
\begin{proof}
In expectation, with $\mu$ as the probability measure associated with the distribution of $\pi$ values, and with $\bar{\pi}:=\int_{[0,1]}\pi d\mu$, since asymptotically $P(e=1)=\bar{\pi}$ we have 
\begin{subequations}
\begin{align}R^2_{t;\pi,e}&:=E(\pi|e=1)-E(\pi|e=0)\\
&=\frac{1}{\bar{\pi}}\int_{[0,1]} \pi\pi d\mu-\frac{1}{1-\bar{\pi}}\int_{[0,1]}(1-\pi)\pi d\mu\\
&=\frac{1-\bar{\pi}}{\bar{\pi}(1-\bar{\pi})}\int_{[0,1]} \pi\pi d\mu-\frac{\bar{\pi}}{\bar{\pi}(1-\bar{\pi})}\int_{[0,1]}(1-\pi)\pi d\mu\\
&=\frac{1}{\bar{\pi}(1-\bar{\pi})}\int_{[0,1]} \pi\pi d\mu-\frac{\bar{\pi}}{\bar{\pi}(1-\bar{\pi})}\int_{[0,1]}\pi d\mu\\
&=\frac{1}{\bar{\pi}(1-\bar{\pi})}\left(E(\pi^2)-E(\pi)^2\right)\\
&=\frac{\sigma^2_\pi}{P(e=1)(1-P(e=1))}.
\end{align} 
\end{subequations}
\end{proof}
Likewise, with $r=(1-\pi)r_0+\pi r_1$, we have the following corollary.
\begin{corollary}[Proportion of Variation Explained]
\label{cor2} Let $R^2_{t;r,d}$ denote Tjur's pseudo-$R^2$ parameter between $r$ and $d$ as defined in Section \ref{methods}. Let $\sigma^2_r$ be the variance of the latent $r$ values, and let $P(d=1)$ be the observed proportion of individuals with $d=1$. Then,
\[R^2_{t;r,e}=\frac{\sigma^2_{r}}{P(d=1)(1-P(d=1)}.\]
\end{corollary}
Within our propensity-risk framework, 
\begin{equation}
    \label{report1}
R^2_{p;e}\leq R^2_{t;\pi,e}
\end{equation}
and
\begin{equation}
    \label{report2}
R^2_{p;d}\leq R^2_{t;r,d}.
\end{equation}
The inequalities of (\ref{report1}) and (\ref{report2}) justify the inequalities of (\ref{piineq}) and (\ref{rineq}), and also the corresponding inequalities of (\ref{fcineq}) and (\ref{scineq}). In words, real models built from incomplete sets of covariates should have $R^2$ coefficients that are less than or equal to idealized models built from a theoretically complete set of measured and unmeasured covariates.
\section{Application}
\label{applications}
Here we analyze observed associations between marijuana use and hard drug use, with and without conditioning on age and gender, and with and without knowledge of relevant $R^2$ coefficients. Based on \cite{Alexander2017} we set
\begin{equation}\label{R2s}R^2_{p;\textrm{marijuana}}=0.30\textrm{~and~}R^2_{p;\textrm{hard drugs}}=0.20.\end{equation} Table \ref{tab2} shows a marginal association between marijuana use and hard drug use. Table \ref{tab1} shows conditional associations between marijuana use and hard drug use. The data of Tables \ref{tab2} and \ref{tab1} were obtained from The Population Assessment of Tobacco and Health (PATH) Study \cite{PATH}. In Appendix \ref{AppData} we provide detailed descriptions of how the data was obtained and how the variables were defined. Here we will interpret the data with the methods of Section \ref{methods}. These example applications are meant to illustrate our methodology, not to support any scientific conclusions.

If we apply Theorem \ref{th1} to the data of Table \ref{tab2} it results in the following estimate. $100\%$ of the population is estimated to have 
\[(\pi=.37,r_0=.2,r_1=.32).\]
In other words, we may assign the population relative risk of $16$ or the population risk difference of $.30$ to each individual. Something similar happens when we apply Corollary \ref{condhomo}. In that case we may assign the conditional relative risk values of Table \ref{tab1} or the conditional risk difference values of Table \ref{tab1} to each individual of a given age (range) and gender.

\begin{table}[ht]
\centering
\caption{A $2\times 2\times 5\times 2$ contingency table relating Marijuana use, hard drug use, age, and gender.}
\label{tab1}
\begin{tabular}{lrccrccr}
&  \multicolumn{2}{c}{Male} && \multicolumn{2}{c}{Female}  \\
\toprule
 && \multicolumn{2}{c}{Marijuana} && \multicolumn{2}{c}{Marijuana}  \\
\cmidrule{3-4} \cmidrule{6-7} 
Age&Hard Drugs&No&Yes&Hard Drugs&No&Yes&\\
\hline
15-25&Yes&8&26&Yes&6&22&\\
&No&613&193&No&695&242&\\
\hline
&& \multicolumn{2}{c}{Marijuana} && \multicolumn{2}{c}{Marijuana}  \\
\cmidrule{3-4} \cmidrule{6-7} 
Age&Hard Drugs&No&Yes&Hard Drugs&No&Yes&\\
\hline
25-35&Yes&11&44&Yes&6&48&\\
&No&403&154&No&459&178&\\
\hline
&& \multicolumn{2}{c}{Marijuana} && \multicolumn{2}{c}{Marijuana}  \\
\cmidrule{3-4} \cmidrule{6-7} 
Age&Hard Drugs&No&Yes&Hard Drugs&No&Yes&\\
\hline
35-45&Yes&6&68&Yes&5&58&\\
&No&323&147&No&400&169&\\
\hline
&& \multicolumn{2}{c}{Marijuana} && \multicolumn{2}{c}{Marijuana}  \\
\cmidrule{3-4} \cmidrule{6-7} 
Age&Hard Drugs&No&Yes&Hard Drugs&No&Yes&\\
\hline
45-55&Yes&11&154&Yes&13&137&\\
&No&317&152&No&403&160&\\
\hline
&& \multicolumn{2}{c}{Marijuana} && \multicolumn{2}{c}{Marijuana}  \\
\cmidrule{3-4} \cmidrule{6-7} 
Age&Hard Drugs&No&Yes&Hard Drugs&No&Yes&\\
\hline
55-65&Yes&6&148&Yes&9&91&\\
&No&267&165&No&321&120&\\
\hline
\end{tabular}
\end{table}


We may make use of the known coefficients of determination that were set in (\ref{R2s}). With those set values we first analyze the marginal data of Table \ref{tab2}. Those set coefficients of determination and the data of Table \ref{tab2} determine the parameters of the optimization problem described in (\ref{transformed}). Without covariates we set  $X=\emptyset$. Application of our LP then produces the solution of Figure \ref{fig1}. 

We have displayed the illustration of Figure \ref{fig1} to highlight the three subpopulations that have emerged. In that figure the area of a dot is proportional to the size of its corresponding subpopulation. The lines of Figure \ref{fig1} represent causal effects. In that figure lines corresponding with subpopulations intersect the left vertical axis at $r_0$ values, and they intersect the right vertical axis at $r_1$ values. We can see in Figure \ref{fig1} how the locations of the dots have explained away some of the marginal association. Overall, the (residual) effects (illustrated with the lines) appear to have become smaller in magnitude. 

\begin{figure}[ht]
\centering
\includegraphics[width=0.8\textwidth]{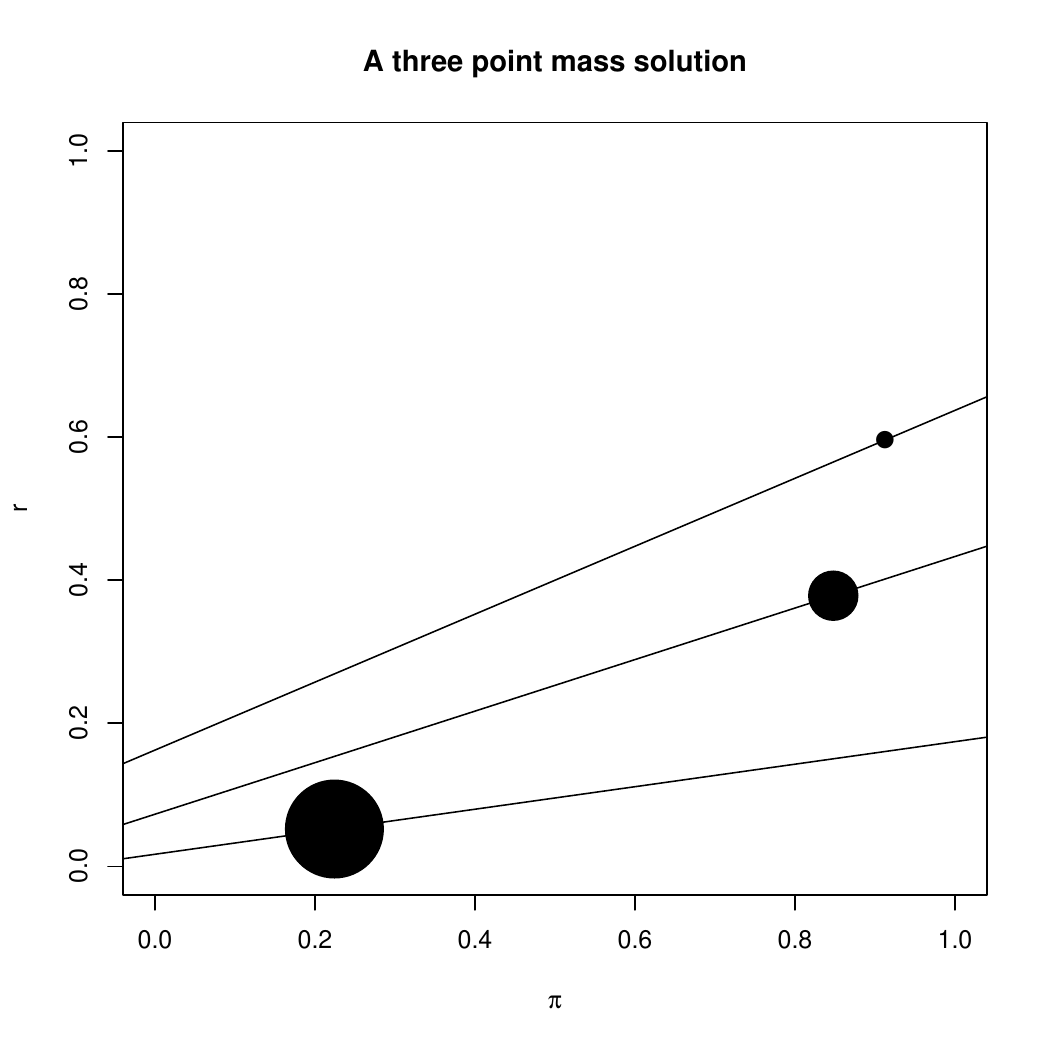}
\caption{Given $R^2_{p;\textrm{marijuana}}=.30$ and $R^2_{r;\textrm{hard drugs}}=.20$ and the data of Table \ref{tab2} we utilize the principle of maximum entropy and solve (\ref{cR}) to estimate that the population consists of $77.5\%$ of individuals with $(\pi=0.22, r_0=0.02,r_1=0.17)$, $20\%$ of individuals with $(\pi=0.85, r_0=0.07,r_1=0.43)$, and the remaining $2.5\%$ of individuals with $(\pi=0.91, r_0=0.16,r_1=0.64)$.} 
\label{fig1}
\end{figure}

Before incorporating knowledge of $R^2$ values into the analysis, we inferred that all individuals had a relative risk of $16$ and a risk difference of $32$. After incorporating knowledge of $R^2$ values into the analysis we infer the following. $77.5\%$ of individuals have a relative risk of $8.5$ and a risk difference of $0.15$. $20\%$ of individuals have a relative risk of $6.1$ and a risk difference of $0.36$. $2.5\%$ of individuals have a relative risk of $4.0$ and a risk difference of $0.48$.  

Something similar happens when we incorporate not just knowledge of $R^2$ values into the analysis, but also measured covariate data. Next we analyze the data of Table \ref{tab1}. Those set coefficients of determination and the data of Table \ref{tab1} determine the parameters of the optimization problem described in (\ref{transformed}). Application of our LP produces the solution of Table \ref{tab3}. 
Remarkably, there are two subpopulations for each level of age and gender. The subpopulation-specific effect sizes $r_1/r_0$ and $r_1-r_0$ are generally smaller, as before, but they also vary depending on age and gender, and generally reflect the magnitudes of the conditional associations in Table \ref{tab1}.
%
\begin{table}[ht]
\centering
\caption{With $R^2_{p;\textrm{marijuana}}=0.30$ and $R^2_{p;\textrm{,hard drugs}}=0.20$ and the data of Table \ref{tab1} our LP has produced the following mixture distribution as the argument maximum solution to the problem in (\ref{transformed}).}
\label{tab3}
\begin{tabular}{llrcclrccr}
&  \multicolumn{3}{c}{Male} && \multicolumn{3}{c}{Female}  \\
\toprule
Age&weight&$\pi$&$r_0$&$r_1$&weight&$\pi$&$r_0$&$r_1$&\\
\hline
15-25&.12&.24&.02&.14&.13&.24&.01&.10&\\
&.004&.74&.04&.14&.01&.74&.04&.11&\\
\hline
25-35&.08&.24&.04&.22&.08&.22&.02&.20&\\
&.01&.79&.04&.29&.01&.76&.04&.26&\\
\hline
35-45&.06&.21&.03&.26&.07&.21&.02&.22&\\
&.02&.81&.04&.38&.02&.79&.04&.31&\\
\hline
45-55&.05&.18&.04&.34&.07&.19&.04&.31&\\
&.04&.86&.11&.56&.04&.86&.09&.54&\\
\hline
55-65&.04&.19&.03&.31&.06&.19&.04&.31&\\
&.04&.86&.06&.52&.02&.85&.07&.52&\\
\hline
\end{tabular}
\end{table}
\section{Discussion}
\label{discussion}
In this paper we have shown how to utilize the principle of maximum entropy to estimate distributions of heterogeneous causal effects. We have seen how to incorporate measured covariate data and knowledge of coefficients of determination within the analysis. Conditioning on a covariate can raise or lower an effect estimate, but adjustment for knowledge of $R^2$ parameters seems to temper the magnitude of any causal interpretation. Since all models are incomplete, typical or common causal inference or interpretation may be generally biased toward more causation. Our introduced methodology can be used to adjust for that bias.

We view this paper as an introductory proof of concept, and there are many ways to extend or expand upon the theoretical results of Section \ref{results}. It's possible that increasing $R^2_{p;e}$ and $R^2_{p;d}$ values can only lower overall magnitudes and measures of overall causality, c.f. \cite{PearlDM}, but it is not yet clear how best to measure or prove that. Also, we have utilized our LP methodology to solve (\ref{transformed}) which approximates (\ref{cR}), but we do not have a proof of correspondence in the limit. We do have empirical evidence to support that correspondence, in the form of LP results converging to the theoretical results that are stated and proven in Section \ref{results}. We also see evidence of convergence when we refine the discretization in the approximation of (\ref{transformed}). This evidence is presented in Figure 2 of Appendix \ref{AppB}. We have witnessed three-point mass distributions as solutions to problems formulated on marginal data, and mixtures of conditional two-point mass distributions as solutions to problems involving measured covariate data. A proof that solutions must be three-point mass distributions or built from two-point mass distributions might lead to an analytic solution of (\ref{cR}) with or without measured covariates. Within the problems described in Section \ref{methods} our constraints may be viewed as known first or second moments. Additional constraints for higher moments or of other types may be added if known.

Our LP utilized tolerance parameters, that we fine tuned to obtain better output. Looser tolerance parameters helped us to find feasible distributions, and then tighter tolerance parameters helped us to produce more accurate output. We plotted solutions and observed the presence of two or three clusters, and then adjusted the approximate results by combining small bits of mass from adjacent subcubes into single chunks of mass, to produce the output presented in Section \ref{applications}. We adjusted by moving mass but only from adjacent subcubes. This technique was tested in special cases and found to agree with the theoretical results of Section \ref{results}. 

We are limited also by the assumptions that went into our model. We utilized a propensity-prognosis model of the data generating process, which many not be appropriate in all settings. In particular, there may be SUTVA violations. Here we are referring to the second assumption of SUTVA as described in \cite[Section 3.2]{Imbens2015}. The second assumption of SUTVA requires that there are not multiple versions of the exposure, so that the potential outcomes are well defined. In our application that may or may not be the case. We do recommend careful specification of times at which exposures and outcomes are measured, and careful specification of times at which propensity and prognosis probabilities are defined. Researchers and analysts are free to choose the times of definition, and we recommend they do so in a way to produce a valid model of the data generating process. SUTVA is an issue because when much time passes between the time of definition and the time of exposure, natural process may give rise to multiple exposures. In the context of our example, chance random events may give rise to both marijuana use and other traits, e.g. alcohol use, which may violate SUTVA. However, our stochastic counterfactuals make SUTVA more believable.

Our results highlight the importance of exposure propensities. In Figure \ref{fig1} and Table \ref{tab3} we see individuals with lower propensity probabilities of the exposure and other individuals with higher propensity probabilities of the exposure. We can enrich the interpretation of common measures of association by reporting $(\pi,r_0,r_1)$ instead of just $r_1/r_0$ or $r_1-r_0$. When contemplating interventions we should consider also the propensity probabilities, and personalized interventions can be tailored for individuals of subpopulations identified with our introduced methodology.

Our numerical approach was not limited by any lack of computing power because our measured covariates had only ten levels. In case of high-dimensional covariate data that may be sparse we can make a ``balance'' assumption \cite[Equation (6)]{Jorke18}. That balance assumption allows us to solve maximum-entropy optimization problems with multivariate, multiple, logistic regression, as described in \cite{Jorke18}, but without $R^2$ parameters. In such an approach the response variable will be $(e,d)$ and the predictors will be the measured covariates in $x$. The fitted values are predicted probabilities for the four joint categories of $(e,d)$. For each individual those probabilities must sum to unity, and they can be transformed into $\{\pi_i,r_{0i},r_{1i}\}$ values, from which we can learn about the distribution of heterogeneous effects.  

If the cell counts are low then there may be uncertainty about the equality constraints of (\ref{cR}), and whether or not they should be required to hold. Next we describe one way to address that issue. We may continue to require that the equality constraints hold as written, but we would repeatedly resample the observed data to produce synthetic samples. On each synthetic sample we would proceed as described in this paper. We would then mix together the resulting solution distributions and interpret the resulting mixture distribution. If the resampled mixture has changed drastically, then we would know that the cell counts are too low. We do not expect our results to be overly sensitive to resampling when cell counts are moderately sized. For details see the content on finite population corrections in \cite{https://doi.org/10.48550/arxiv.2309.02621}.

In conclusion, we highlight how this paper is ultimately about unmeasured covariates. Unmeasured covariates can be handled in a principled manner with the principle of maximum entropy, as we have shown here. It is reasonable to select covariates that might cause $e$ or $d$ \cite{VanderWeele2011}, and we may imagine having done so hypothetically. We may not know the particulars of those unmeasured covariates, but we may reasonably bound related $R^2$ parameters, and those bounds can serve as a relevant summary of those unmeasured covariates; c.f. \cite{KOA,Knaeble2023}. As seen here in Section \ref{applications}, those bounds can have a greater impact on causal interpretation than measured covariates. We have seen how those bounds and the principle of maximum entropy partition an apparently homogeneous population into two or three differing subpopulations. The result is a principled middle-way for causal inference, between the extremes of full causal interpretation of an observed association and skeptical insistence that the association is entirely spurious. The differing subpopulations explain some of the observed association, but some residual amount of the observed association admits a causal interpretation, and the interpretation is enriched, as it points toward more personalized interventions.
%
\bibliographystyle{plain}

\newpage
\appendix
\section{Appendix}
\subsection{Data Preparation}
\label{AppData}
The data of Section \ref{applications} was obtained from The Population Assessment of Tobacco and Health (PATH) Study \cite{PATH}. From ICPSR 36498 we downloaded DS1001 Wave 1: Adult Questionnaire Data with Weights. We utilized raw (unweighted) data on six variables. The first variable was (our exposure $e$) R01\_AX0085 
(Ever used marijuana, hash, THC or grass). Respondents answered yes or no. The outcome variable (our disease $d$) was derived from three variables: R01\_AX0220\_01 (Ever used substance: Cocaine or crack), R01\_AX0220\_02 (Ever used substance: Stimulants like methamphetamine or speed), and R01\_AX0220\_03 (Ever used substance: Any other drugs like heroin, inhalants, solvents, or hallucinogens). We recorded $d=1$ if a respondent answered yes to any of the questions R01\_AX0220\_01, R01\_AX0220\_02, or R01\_AX0220\_03, and we recorded $d=0$ otherwise. The two measured covariates were
R01R\_A\_AGECAT7 (Age range when interviewed (7 levels)) and R01R\_A\_SEX (Gender from the interview (male or female)). We excluded the two oldest age levels, restricting attention to individuals between the ages of $15$ and $65$. We did not weight the data.
\subsection{Convergence}
\label{AppB}
We repeated some trials of our LP on the problems described in (\ref{basic}) and (\ref{c}) with data from Table \ref{tab1}. We compared the numerical results with the theoretical results of Theorem \ref{th1} and Corollary \ref{condhomo}. We increased the $m$ parameter of (\ref{transformed}) from $25$ to $95$ in increments of $5$, and plotted $m$ and the resulting max entropy outputs of the LP on the plots below. The height of the dots is the max entropy according to the LP. The height of the horizontal line is the max entropy according to Theorem \ref{th1} and Corollary \ref{condhomo}.  
\begin{figure}[ht]
\centering
\includegraphics[width=0.8\textwidth]{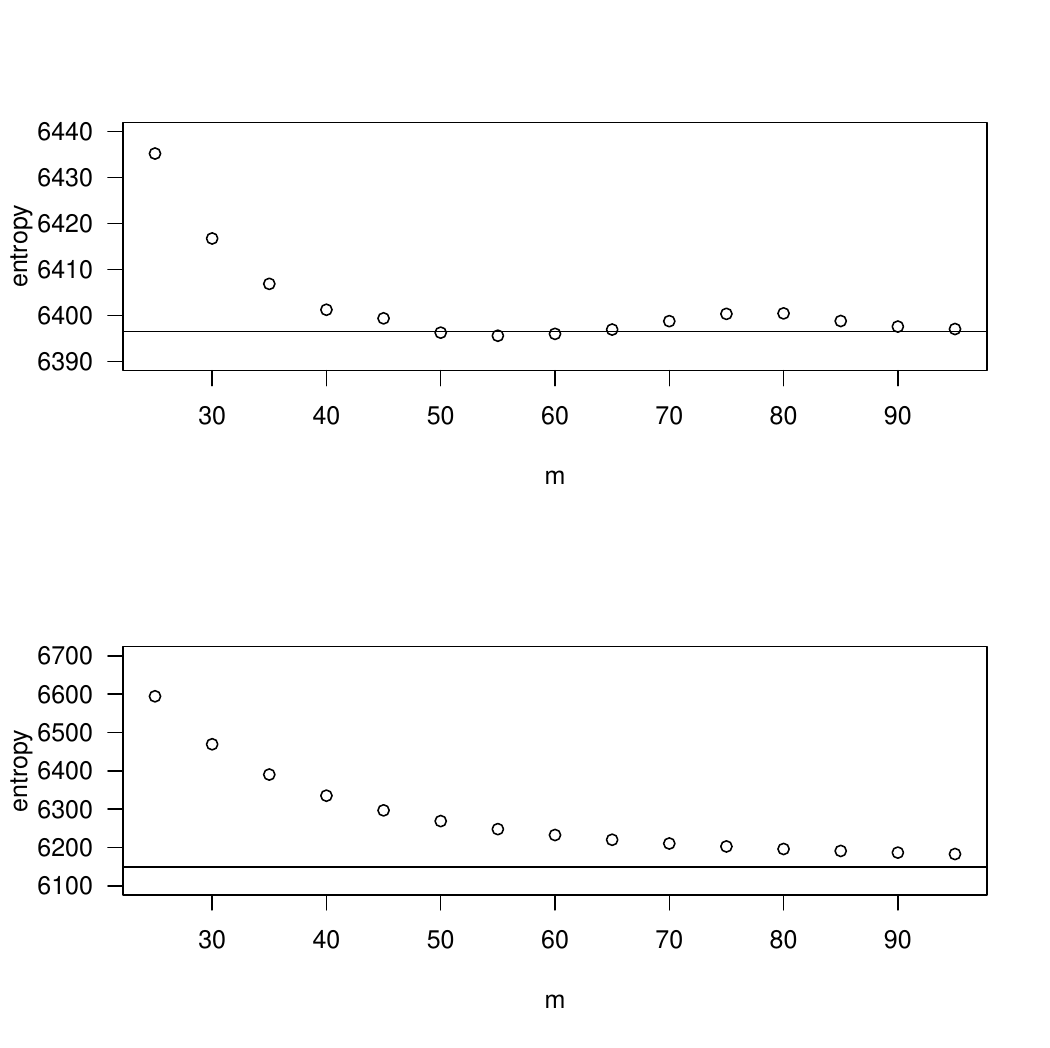}
\caption{Convergence of LP max entropy to the theoretical maximum entropy of problem (\ref{basic}) [top] and problem (\ref{c}) [bottom] with data from Table \ref{tab1}.} 
\label{fig2}
\end{figure}
\end{document}